\documentclass{article}
\usepackage{ijcai19}

\usepackage{times}
\usepackage{soul}
\usepackage{url}
\usepackage[hidelinks]{hyperref}
\usepackage[utf8]{inputenc}
\usepackage[small]{caption}
\usepackage{graphicx}
\usepackage{amsmath}
\usepackage{booktabs}
\usepackage{algorithm}
\usepackage{algorithmic}
\usepackage{subcaption}

\usepackage{amssymb,amsthm,amsmath}
\usepackage{color}
\usepackage{lipsum}
\usepackage{graphicx}

\RequirePackage[textsize=scriptsize]{todonotes}
\usepackage{enumitem}

\newtheorem{theorem}{Theorem}
\newtheorem{corollary}{Corollary}
\newtheorem{lemma}{Lemma}

\newcommand{\E}{\mathbb{E}}
\DeclareMathOperator{\var}{Var}
\newcommand{\indep}{\bot\!\!\!\!\bot}

\newcommand{\R}{\mathbb{R}}
\newcommand{\N}{\mathcal{N}}

\title{Three-quarter Sibling Regression for Denoising Observational Data}

\author{
Shiv Shankar$^1$\footnote{Contact Author}\and
Daniel Sheldon$^{1,2}$\and
Tao Sun$^{3}$\and
John Pickering$^{4}$\And
Thomas G. Dietterich$^{5}$ \\
\affiliations
$^1$University of Massachusetts Amherst\\
$^2$Mount Holyoke College\\
$^3$Amazon\\
$^4$University of Georgia\\
$^5$Oregon State University\\
\emails
\{sshankar, sheldon\}@umass.edu,
suntao.st@gmail.com, 
pick@discoverlife.org,
tgd@oregonstate.edu
}

\begin{document}

\maketitle

\begin{abstract}
Many ecological studies and conservation policies are based on field observations of species, which can be affected by systematic variability introduced by the observation process.
A recently introduced causal modeling technique called ``half-sibling regression'' can detect and correct for systematic errors in measurements of multiple independent random variables.
However, it will remove intrinsic variability if the variables are dependent, and therefore does not apply to many situations, including modeling of species counts that are controlled by common causes.
We present a technique called ``three-quarter sibling regression'' to partially overcome this limitation. It can filter the effect of systematic noise when the latent variables have observed common causes.
We provide theoretical justification of this approach, demonstrate its effectiveness on synthetic data, and show that it reduces systematic detection variability due to moon brightness in moth surveys.

\end{abstract}

\section{Introduction}

Observational data is increasingly important across a range of domains and may be affected by measurement error. Failure to account for measurement error may lead to incorrect inferences.
For example, instrument noise in telescope data can prevent detections of exoplanet transits~\cite{Scholkopf15}; under-reporting of drug use may lead to biased public health decisions~\cite{adams2019learning}; label noise in machine learning training data may lead to suboptimal models~\cite{nettleton2010study,frenay2014classification}; and imperfect detection in ecological surveys may lead to incorrect conclusions about species populations and demographics without the proper modeling~\cite{MacKenzie2002,hutchinson2017species}.
It is therefore important to develop statistical approaches to model and correct for measurement errors.

This paper is motivated by the analysis of ecological survey data.
Surveys conducted by humans provide information about population sizes and dynamics for scientific understanding of animal populations and for setting conservation policies.
However, humans may fail to detect a species for a range of reasons, including animal behaviors, weather, and observer skill.
Making decisions to sustain animal populations requires correctly interpreting survey data with these sources of error.

\begin{figure*}[t]
  \begin{minipage}[b]{0.34\linewidth}
    \begin{subfigure}[b]{0.48\linewidth}
      \centering
      \includegraphics[width=0.9\linewidth]{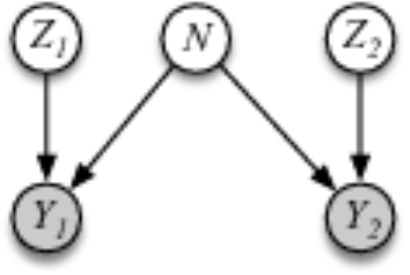}
      \caption{\label{fig:half-sibling}}
    \end{subfigure}
    \hfill
    \begin{subfigure}[b]{0.48\linewidth}
      \centering
      \includegraphics[width=0.9\linewidth]{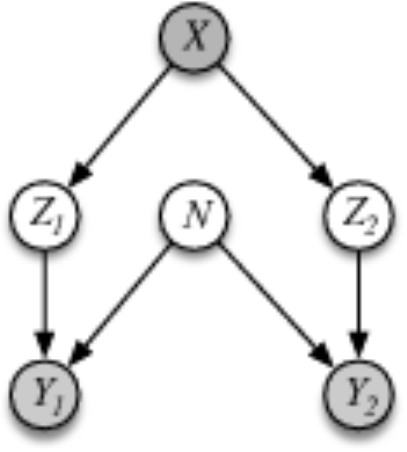}
      \caption{\label{fig:three-quarter-sibling}}
    \end{subfigure}
    \addtocounter{figure}{-1} 
    \captionof{figure}{(a) Half-sibling regression. (b) Three-quarter sibling regression.}
  \end{minipage}
  \hfill
  \begin{minipage}[b]{0.63\linewidth}
    \includegraphics[width=\linewidth]{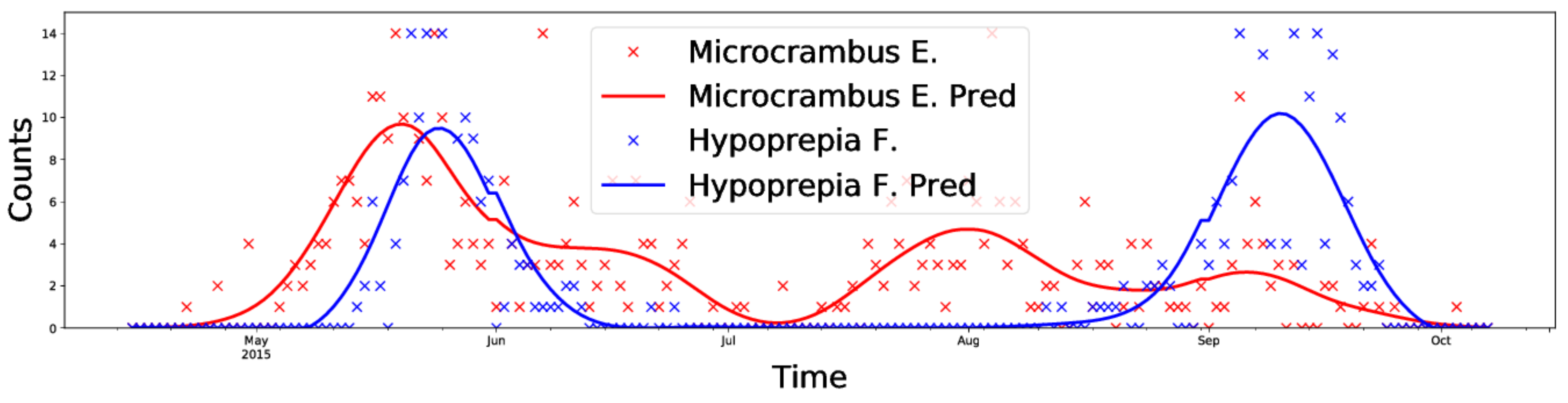}
    \captionof{figure}{\label{fig:moth-example}
      Seasonal patterns of \emph{Microcrambus elegans} and \emph{Hypoprepia fucosa} induce a correlation in their abundance, even though there is no direct causal relationship.}
  \end{minipage}
\end{figure*}

There are several lines of existing work for handling measurement error.
Prior work in the ecology literature has focused on explicitly modeling the detection process using latent-variable models~\cite{MacKenzie2002,Royle2004} \nocite{Dail2011}.
These generally assume some repetition in surveys, which helps distinguish variability in the detection process from intrinsic variability in animal counts.
As citizen-science data increases in importance, another line of work has sought to model observer variability, either directly in latent variable models~\cite{yu2010modeling,yu2014latent,hutchinson2017species}, or as separate metrics~\cite{yu2014clustering,kelling2015can}, for example, to help explain observer-related variability in regression models.
Latent-variable models are also used to model observation error in public health and machine learning with label noise~\cite{frenay2014classification}.
Under certain parametric assumptions, the parameters of latent-variable models for measurement error are identifiable---meaning it is possible to correctly attribute variability to the detection process as opposed to the underlying process---even \emph{without} repeated surveys~\cite{lele2012dealing,solymos2016revisiting,adams2019learning}. \nocite{knape2016assumptions,Knape2015}
However, there is also vigorous debate in the same literature about the assumptions required for identifiability.

Causal modeling is an appealing alternative to parametric latent-variable models.
\cite{Scholkopf15} presented a method called ``half-sibling regression'' that uses causal independence assumptions to detect and remove systematic measurement error.
Independence assumptions are significantly easier to reason about than specific parametric assumptions about the relationships between variables.
The basic idea is to examine simultaneous measurements of multiple quantities that are known \emph{a priori} to be independent.
Any dependence in the measurements must be due to measurement error, which can be quantified and partially removed. 
Although appealing, half-sibling regression only applies when the hidden variables of interest are independent.
We wish to apply similar reasoning to remove systematic noise in surveys of many species.
If counts for all species are lower than expected on a given day, this may well be due to detectability and not actual changes in the populations.
However, there are many common factors such as habitat and time of year that influence the true counts of different species, so these variables are \emph{not} independent, and half-sibling regression does not apply. 

We introduce a method called ``three-quarter sibling regression'' (3QS-regression) that extends half-sibling regression and can remove systematic errors in measurements of quantities that have observed common causes.
Three-quarter sibling regression is derived from a causal model and has a simple interpretation when applied to species counts: the residuals of one species with respect to predictions made using the causal variables are used to correct counts of another species.
We prove bounds on the ability of 3QS-regression to reduce measurement error under different assumptions.
We apply 3QS-regression to a moth survey data set and show that it effectively reduces measurement error caused by moon brightness, which makes it harder to attract moths to artificial lights at survey sites.

\section{Background: Half-Sibling Regression}


The model motivating half-sibling regression is shown in Figure~\ref{fig:half-sibling}.
Here, $Z_1$ and $Z_2$ are two quantities of interest that are known \emph{a priori} to be independent.
In the application of~\cite{Scholkopf15}, these are the brightness values of two distant stars.
The variables $Y_1$ and $Y_2$ are measurements of $Z_1$ and $Z_2$ that are affected by a common noise mechanism $N$, for example, jitter in a telescope. 
In this model, any dependence between $Y_1$ and $Y_2$ must be introduced by the noise mechanism.
The goal of half-sibling regression is to estimate one target value $Z_1$.
It does so using the estimator
$$
\hat{Z}_1 = Y_1 - \E[Y_1 \mid Y_2]
$$
which has the interpretation of subtracting from $Y_1$ the portion of $Y_1$ that can be predicted by $Y_2$. 
In practice, $Y_2$ may be vector valued (e.g., the brightness measurements of \emph{many} other stars), and the graphical model need not follow the exact form given in Figure~\ref{fig:half-sibling} as long as $Y_2 \indep Z_1$ but
$Y_2$ is \emph{not} independent of $Y_1$ given $N$, so that $Y_2$ contains some ``signature'' of the noise process that produces $Y_1$ from $Z_1$.


%

\subsection{Throwing out the Baby with the Bathwater} 
The key limitation of half-sibling regression for our purposes occurs when $Z_1$ and $Z_2$ are \emph{not} independent a priori, for example, due to a common cause.
In our application $Z_1$ and $Z_2$ will represent the counts of different species in a survey area, which are affected by common factors such as the habitat and time of year (see Fig.~\ref{fig:moth-example}).
In this setting, $N$ represents factors that affect detectability and are shared across species. In our moth example, moon brightness is key factor, which induces a \emph{second} source of dependence between moth counts.
More generally, $N$ may include transient factors such as weather, or observer attributes such as skill.


If $Z_1$ is \emph{not} independent of $Z_2$, there are two sources of dependence between $Y_1$ and $Y_2$: (1) the \emph{a priori} dependence induced by $Z_1$ and $Z_2$, and (2) the dependence introduced by the common noise mechanism $N$. 
If half-sibling regression is applied in this case, the correction term $\E[Y_1 \mid Y_2]$ will remove some of the true signal, or, as described by \cite{Scholkopf15}, it will throw out (some of) the baby with the bathwater.
The variable $Y_2$ contains information about the noise mechanism, but it is unclear whether or how this can be teased apart from the \emph{a priori} dependence.

\section{Three-Quarter Sibling Regression}



We consider the model shown in Figure~\ref{fig:three-quarter-sibling}. This extends half-sibling regression by adding the observed variable $X$, which is causal for both $Z_1$ and $Z_2$. We call $X$ the \emph{process covariates}. A key assumption of this model is that $Z_1 \indep Z_2 | X$. This implies two things. First, there is no direct causal link between $Z_1$ and $Z_2$, which is appropriate for most species pairs in our application to ecological surveys. Second, there are no \emph{unobserved} common causes of $Z_1$ and $Z_2$. It is up to the modeler to judge the validity of this assumption, which is standard in causal modeling.\footnote{If either assumption fails it will lead again to the problem of (partially) throwing out the baby with the bathwater.} In our moth survey application, time-of-year will be the single process covariate, and it is reasonable to assume that there are no other common causes. 

In this model, $Y_1$ and $Y_2$ are now ``three-quarter'' siblings: they share one parent and their unshared parents are siblings.\footnote{This term is most commonly used to describe animal kinship relationships} Mathematically, the key assumption is $Z_1 \indep Y_2 \mid X$.
One of our results will also require $N \indep X$. Again, although we are motivated by the particular generative model of Figure~\ref{fig:three-quarter-sibling}, the method applies to any model that meets these assumptions.
For example, the symmetry of the model and presence of $Z_2$ is not required --- $Y_2$ can be \emph{any} variable that contains some information about the noise mechanism and is conditionally independent of $Z_1$ given $X$.

\subsection{Estimator}
The ``three-quarter sibling regression'' estimator or \emph{3QS-estimator} is:
\begin{equation}
\hat{Z}_1 = Y_1 - \E\big[\underbrace{Y_1 - \E[Y_1 \mid X]}_{\text{residual}} \bigm|  X, Y_2 \big]
\label{eq:estimator}
\end{equation}
This is equal to $Y_1$ minus a correction term. The quantity $Y_1 - \E[Y_1 \mid X]$ is the residual after predicting $Y_1$ using $X$ alone. This residual is partly due to intrinsic variability in $p(Z_1 \mid X)$, which we want to preserve, and partly due to the effect of noise when producing the measurement $Y_1$. The variables $X$ and $Y_2$ should not be predictive of the intrinsic variability but may be predictive of the measurement noise. We subtract the portion of the residual that can be predicted using $X$ and $Y_2$ in order to correct $Y_1$ towards $Z_1$.

\noindent Next we prove some results for 3QS analogous to ones in~\cite{Scholkopf15}.
Under the model of Figure~\ref{fig:three-quarter-sibling} and an additional mild assumption, the 3QS-estimator more accurately approximates $Z_1$ than $Y_1$ does. 
\begin{theorem}
  Assume $Z_1 \indep Y_2 \mid X$ and $\E[Y_1 \mid X] = \E[Z_1 \mid X]$.
 Then
\[
\E\big[ (\hat{Z}_1  - Z_1 )^2 \big] \leq
\E\big[ ( Y_1 - Z_1 )^2 \big].
\]
\label{thm:non-additive}
\end{theorem}
\vspace{-0.4cm}
\noindent The additional assumption that $\E[Z_1 \mid X] = \E[Y_1 \mid X]$ is satisfied by an additive zero-mean noise model, as well other models. Note that this result does \emph{not} require $N \indep X$, as implied by Figure~\ref{fig:three-quarter-sibling}.
Before proving the theorem, we give an alternative expression for
$\hat{Z}_1$ that is useful for analysis.

\begin{lemma} An equivalent expression for $\hat{Z}_1$ is
\begin{equation}
  \hat{Z}_1 = Y_1 - \E[Y_1 \mid X, Y_2] + \E[Y_1 \mid X]
  \label{eq:estimator-equiv}
\end{equation}
\end{lemma}
\begin{proof}
From Eq.~\eqref{eq:estimator} we have
\begin{align*}
\hat{Z}_1
&= Y_1 - \E\big[Y_1 - \E[Y_1 \mid X] \bigm|  X, Y_2 \big] \\
&= Y_1 - \E[Y_1 \mid X, Y_2] +  \E \big[\E[Y_1 \mid X] \bigm|  X, Y_2 \big]\\
&=  Y_1 - \E[Y_1 \mid X, Y_2] +  \E[Y_1 \mid X]
\end{align*}
The last line holds because $\E[Y_1 \mid X]$ is a deterministic function of $(X, Y_2)$. For any random variable $U$ and deterministic function $g$ it is the case that $\E[g(U)\mid U] = g(U)$. 
\end{proof}

\begin{proof}[Proof of Theorem~\ref{thm:non-additive}]
First, note that $\E[Y_1 \mid X] = \E[Z_1 \mid X]$ implies $\E[Y_1] = \E[Z_1]$. Then
\begin{align}
  \E\Big[&\big(Z_1 - {Y}_1\big)^2\Big] \notag \\
&= \E\Big[ \big((Z_1 - \E Z_1 \big) - (Y_1 - \E Y_1) \big)^2 \Big] \label{eq:first} \\
&= \E\Big[ \big(Z_1 - Y_1 - \E[Z_1 - Y_1] \big)^2\Big] \notag \\
&\geq \E\Big[ \big(Z_1 - Y_1 - \E[Z_1 - Y_1 \mid Y_2, X] \big)^2\Big] \label{eq:condition} \\
&= \E\Big[ \big(Z_1 - Y_1 - \E[Z_1 \mid X] + \E[Y_1 \mid Y_2, X] \big)^2\Big] \label{eq:indep} \\
&= \E\Big[ \big(Z_1 - (Y_1 - \E[Y_1 \mid Y_2, X] + \E[Y_1 \mid X])\big)^2\Big] \label{eq:same-mean} \\
&= \E\Big[ \big(Z_1 - \hat{Z}_1\big)^2\Big] \notag
\end{align}
In Eq.~\eqref{eq:first}, we subtracted $\E Z_1$ and added $\E Y_1$, which are equal.
In Eq.~\eqref{eq:condition}, we used the fact that conditional variance is no more than total variance.
In Eq.~\eqref{eq:indep}, we used the fact that $Z_1 \indep Y_2 \mid X$.
In Eq.~\eqref{eq:same-mean}, we used $\E[Z_1 \mid X] = \E[Y_1 \mid X]$ and reordered terms.
\end{proof}

\subsection{Additive Noise Model}
Now, further assume the following additive form for $Y_1$:
\begin{equation}
  Y_1 = Z_1 + f(N).
  \label{eq:additive}
\end{equation}
Here $Z_1$ is the ``true'' value and $f(N)$ an additive error term due to $N$. More generally, we could have $Y_1 = \phi(Z_1) + f(N)$ where $\phi$ is an unknown transformation. Since we can never learn such a transformation from observations of $Y_1$,  we assume  the form in Eq.~\eqref{eq:additive}, which amounts to reconstructing the hidden variable \emph{after} transforming it to the same units as $Y_1$.

Under the additive noise model, it is possible to quantify the error of the 3QS-estimator for reconstructing $Z_1$.
\begin{theorem}
Assume $Z_1 \indep Y_2 \mid X$ and $N \indep X$. Under the additive model of Eq.~\eqref{eq:additive}, we have
\[
\E\big[ \big(\hat{Z}_1 - (Z_1 + \E[f(N)])\big)^2\big] = \var\big[ f(N) \mid X, Y_2 \big].
\]
\label{thm:additive}
\end{theorem}
\vspace{-0.2cm}
\begin{proof}
  \begin{align}
\hat{Z}_1
&= Y_1 - \E[Y_1 \mid X, Y_2] + \E[Y_1 \mid X] \notag \\
& \label{line2} \begin{aligned}
    = Z_1 &+ f(N) - \big(\E[Z_1 | X, Y_2] + \E[f(N) | X, Y_2]\big) \\
    &+ \big(\E[Z_1 \mid X] + \E[f(N) \mid X] \big)
  \end{aligned} \\
& \label{line3} \begin{aligned}
  = Z_1 &+ f(N) - \big(\E[Z_1 | X] + \E[f(N) | X, Y_2]\big) \\
  &+ \big(\E[Z_1 \mid X] + \E[f(N)] \big)
  \end{aligned} \\
&= Z_1 + \big(f(N) - \E[f(N) \mid X, Y_2] \big) + \E[f(N)] \label{line4}
\end{align}
Eq.~\eqref{line2} uses the additive expansion of $Y_1$ three times. Eq.~\eqref{line3} uses the facts that $Z_1 \indep Y_2 \mid X$ and $N \indep X$. Eq.~\eqref{line4} rearranges. Then, rearranging,
\[
\hat{Z}_1 - (Z_1 + \E[f(N)]) = f(N) - \E[f(N) \mid X, Y_2].
\]
Therefore,
\begin{align*}
  \E\big[ \big(\hat{Z}_1 \!-\! (Z_1 \!+\! \E[f(N)])\big)^2
    &= \E\big[ \big(f(N) \!-\! \E[f(N) | X, Y_2]\big)^2\big] \\
    & = \var[f(N) \mid X, Y_2].
\end{align*}
\end{proof}
\noindent
Theorem~\ref{thm:additive} says that it is possible to reconstruct $Z_1$---up to a constant additive offset equal to the mean of the measurement error---with squared error equal to the conditional variance of the measurement error given the observed variables $X$ and $Y_2$. If the measurement error is completely determined by the observed variables, then $Z_1 + \E[f(N)]$ is reconstructed exactly. 

\begin{corollary}
  If there is a function $\psi$ such that $f(N) = \psi(X, Y_2)$, then $\hat{Z}_1 = Z_1 + \E[f(N)]$.
  \label{cor:exact}
\end{corollary}

\begin{proof}[Proof]
  In this case
  $\var\big[f(N) \bigm| X, Y_2 \big] = \var\big[\psi(X, Y_2) \bigm|
    X, Y_2 \big] = 0 $, which implies that $\hat{Z}_1 \!- (Z_1 + \E[f(N)]) = 0$.
\end{proof}

\subsubsection{Implementation}
In practice, the conditional expectations in the 3QS-estimator are unknown, but can be replaced by regression models. The 3QS-estimator used in practice is:
\[
\hat{Z}_1 = Y_1 - \hat{\E}\big[Y_1 - \hat{\E}[Y_1 \mid X] \bigm|  X, Y_2 \big]
\]
where $\hat{\E}[B \mid A]$ is a regression model trained to predict $B$ from $A$.

In our application, the $Y_i$ variables are symmetric and represent counts of different species. We first fit a regression model $\hat{\E}[Y_i \mid X]$ using only process covariates for each species. Then, define $R_i = Y_i - \hat{\E}[Y_i \mid X]$ to be the residual for species $i$, and let $R_{-i}$ be the vector of residuals for all other species. We use the estimator
\[
\hat{Z}_i = Y_i - \hat{\E}\big[R_i \mid R_{-i}].
\]
This has a very simple interpretation. We subtract the portion of the residual for the target species that can be predicted using the residuals of the other species. This is a special case of the 3QS-estimator where the regression model using $(X, Y_{-i})$ as predictors is parameterized as a function of only the residual $Y_{-i} - \hat{\E}[Y_{-i} \mid X]$.


\section{Experiments}
In this section we experimentally demonstrate the ability of
3QS-regression to remove systematic measurement error in the presence of common causes using first synthetic data and then a moth survey data set from the Discover Life project.


 


\subsection{Synthetic Experiments}


\begin{figure}[t]
\begin{subfigure}{0.49\linewidth}
  \centering
  \includegraphics[width=\linewidth]{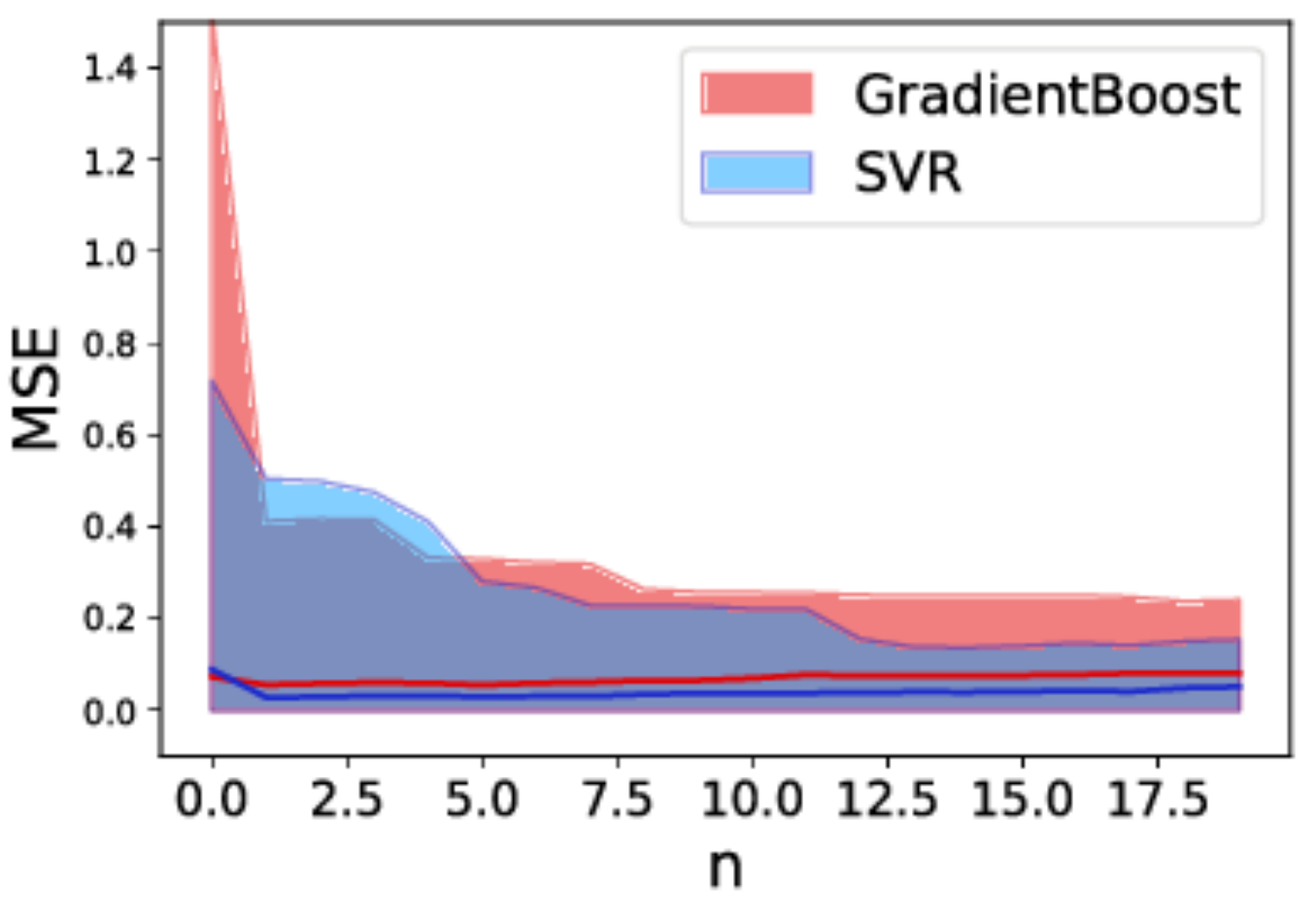}
  \caption{\label{fig:numY} }
\end{subfigure}
\begin{subfigure}[hb]{0.49\linewidth}
  \centering
  \includegraphics[width=\linewidth]{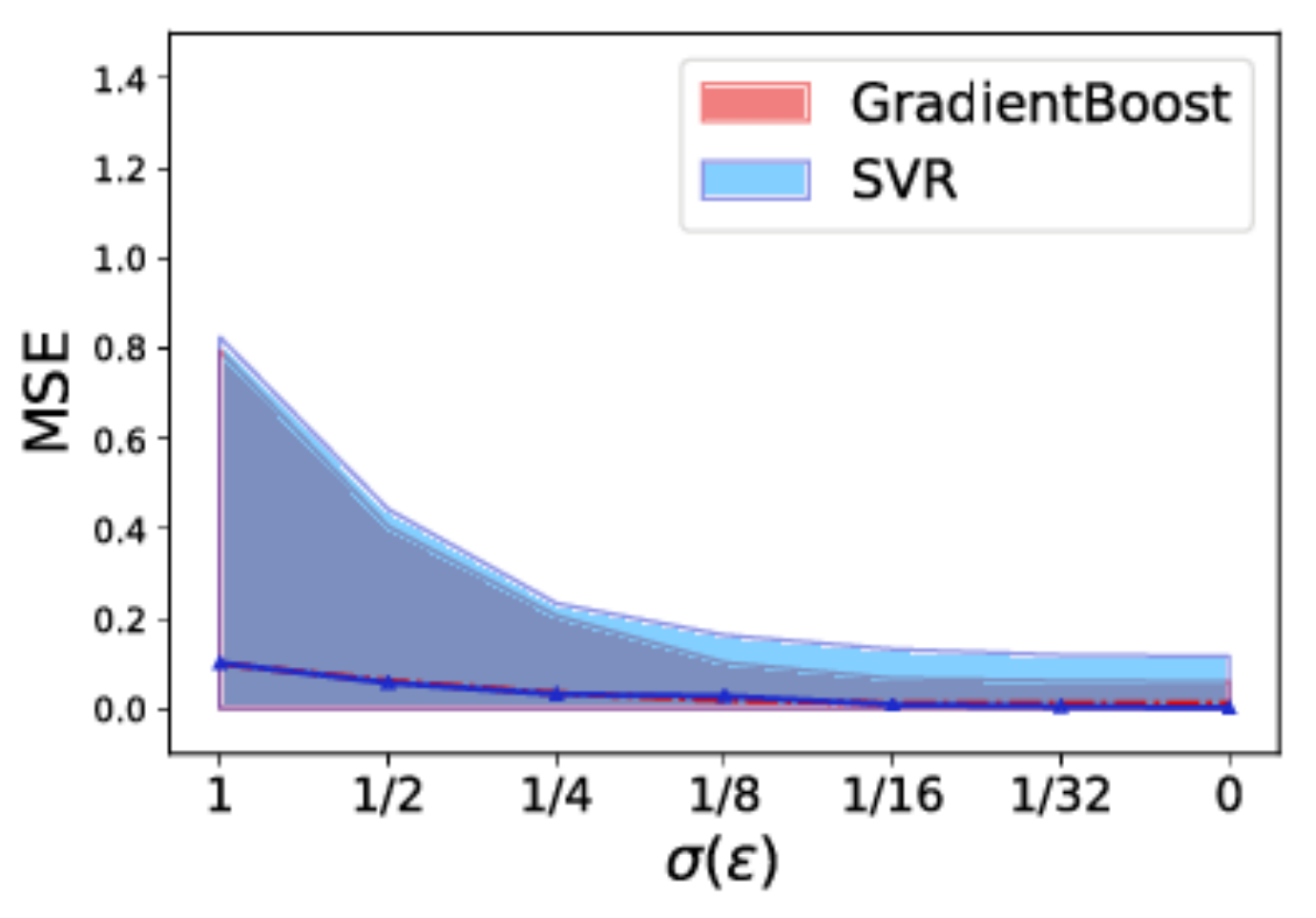}
  \caption{\label{fig:noiseY} }
\end{subfigure}
\begin{subfigure}{0.49\linewidth}
  \centering
  \includegraphics[width=\linewidth]{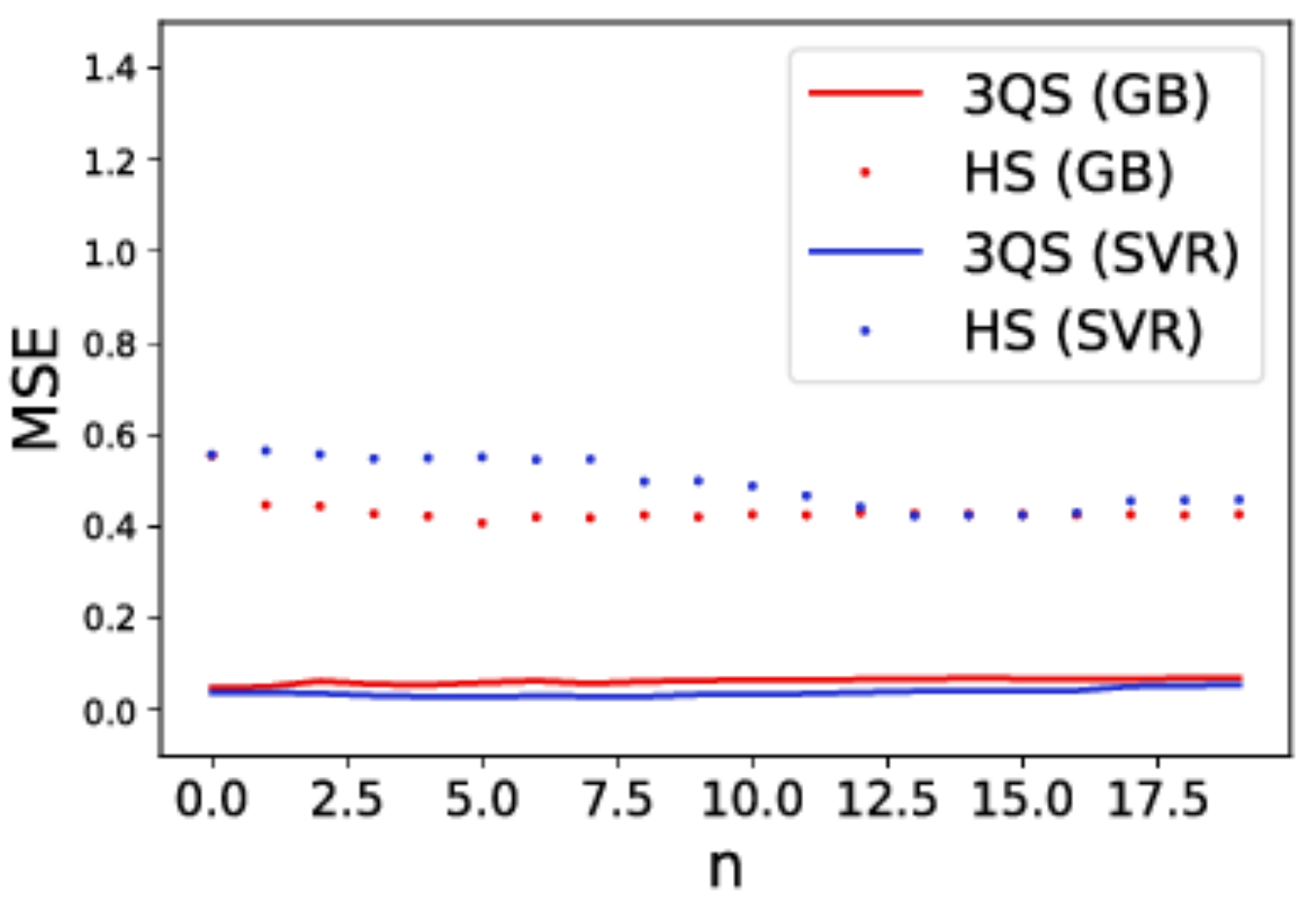}
  \caption{\label{fig:numYcomp} }
\end{subfigure}
\begin{subfigure}[hb]{0.49\linewidth}
  \centering
  \includegraphics[width=\linewidth]{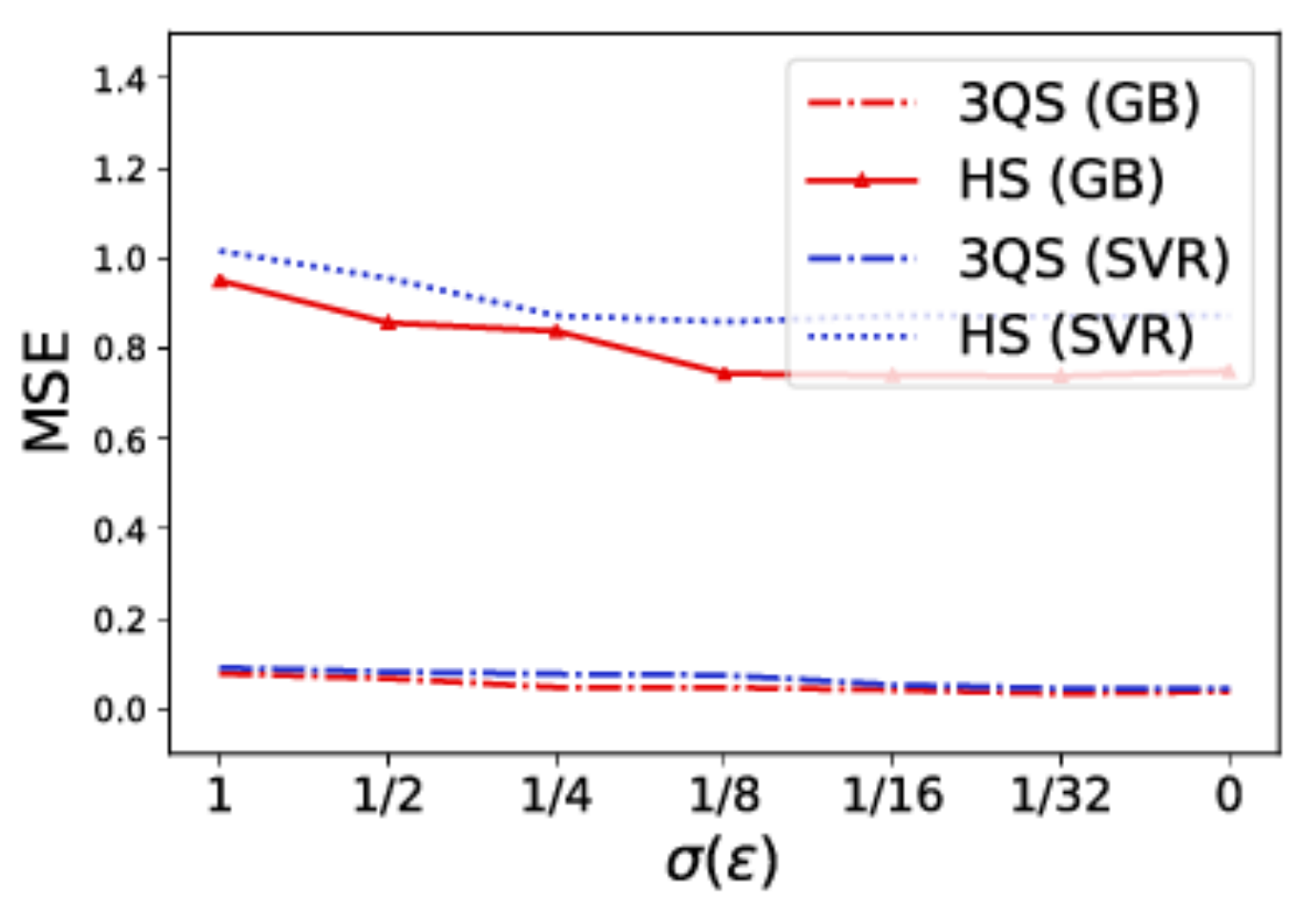}
  \caption{\label{fig:noiseYcomp} }
\end{subfigure}
\caption{Results on synthetic data. (a)~MSE vs species count $n$, (b)~MSE vs noise deviation $\sigma_\epsilon$, (c)~HS and 3QS MSE vs. $n$, (d)~HS and 3QS MSE vs. $\sigma_\epsilon$.}
\end{figure}

We first conducts experiments using synthetic data, for which the true
value of $Z_1$ is known.
This allows us to measure the ability of 3QS-regression to recover
$Z_1$ and evaluate its performance in settings where the auxiliary
measurements $Y_2 = Y_{-1}$ contain a varying amount of information about the noise mechanism.

\subsubsection{Methodology}
Our synthetic data generation is patterned off the experiments
of~\cite{Scholkopf15} for half-sibling regression, but extended to
include process covariates.
We simulated $n$ species (indexed by $i$) whose occurrences $Y_i$ are
determined as the following function of a process covariate $X \in \R$ and
noise variable $N \in \R$: 
\[
  Y_i = \underbrace{w^{(i)}_X X}_{Z_i} + \underbrace{g_i(w^{(i)}_N
    N)}_{f_i(N)} + \epsilon.
\]
The coefficients $w^{(i)}_X$ and $w^{(i)}_N$ are
drawn uniformly from $[-1, 1]$ for each species, and control the
relationship between the species occurrence and the process and noise
variables. The function $g_i$ is a sigmoid function with randomized
parameters (to control the inflection points, etc.). Finally $\epsilon \sim \N(0, \sigma^2_\epsilon)$ is
independent noise. 


We conduct two experiments to simulate decreasing the uncertainty about $f_1(N)$ given $Y_{-1}$. In the first case we set $\sigma^2_\epsilon = 0$ and increase the number of species, each with its own noise function $f_i$. As the number of species $n$ increases, we get more predictors of the error $f_1(N)$.
In the second experiment we fix $n = 2$ and $f_1=f_2$ (i.e the effect of noise on both species is exactly the same) while varying the noise $\sigma^2_\epsilon$. In either case, when the conditional variance of $f_1(N)$ given $(X,Y_{-1})$ reduces, Theorem~\ref{thm:additive} predicts a more accurate reconstruction.



For these experiments we used the alternate form of the estimator
given in Eq.~\ref{eq:estimator-equiv}, which means we fit regression
models $\hat{\E}[Y_i | X, Y_{-i}]$ and $\hat{\E}[Y_i | X]$ and
computed $\hat{Z}_i$ according to Eq.~\ref{eq:estimator-equiv}. The estimators could be shifted from $Z_i$ by any constant offset (cf. Theorem~\ref{thm:additive}), so we centered them to have zero-mean to match $Z_i$.
We conduct our experiments with two different regression algorithms---support-vector regression (SVR) and gradient boosted regression trees,
using implementations and default settings from the scikit-learn
package.\footnote{https://pypi.org/project/scikit-learn/}
Following \cite{Scholkopf15}, we create 20 different fixed instances
and measure the mean squared error (MSE) of our denoised counts
$\hat{Z}_i$ against true value $Z_i$.

\paragraph{Results.}
Figures ~\ref{fig:numY} and \ref{fig:noiseY} show the reconstruction
error as a function of the number of species $n$ and noise variance
$\sigma^2_\epsilon$, respectively. The figures plot mean MSE and
error across the runs.
Fig.~\ref{fig:numY} shows that increasing the number of species in
$Y_{-i}$ causes error to decrease with both regression methods.
This is expected, because with more predictors that are
correlated with $f(N)$, we can learn a better model for the systematic error.
Fig.~\ref{fig:noiseY} shows that as {\small $\sigma^2_\epsilon \to 0$} the
error also tends towards zero. This is also expected, because, in this case, $f(N)$ becomes a deterministic function of $X,Y_2$.
Figs.~\ref{fig:numYcomp} and~\ref{fig:noiseYcomp} compare the mean MSE of half-sibling (HS) regression against our method (3QS). HS regression performs poorly due to the common dependence of  $Y_i$ variables on $X$.




\subsection{Discover Life Moth Observations}
\label{s:moths}

\begin{figure*}[t]
  \centering
  \hspace*{-0.5cm}
    \includegraphics[width=0.7\textwidth, trim={0 2cm 0 0},clip]{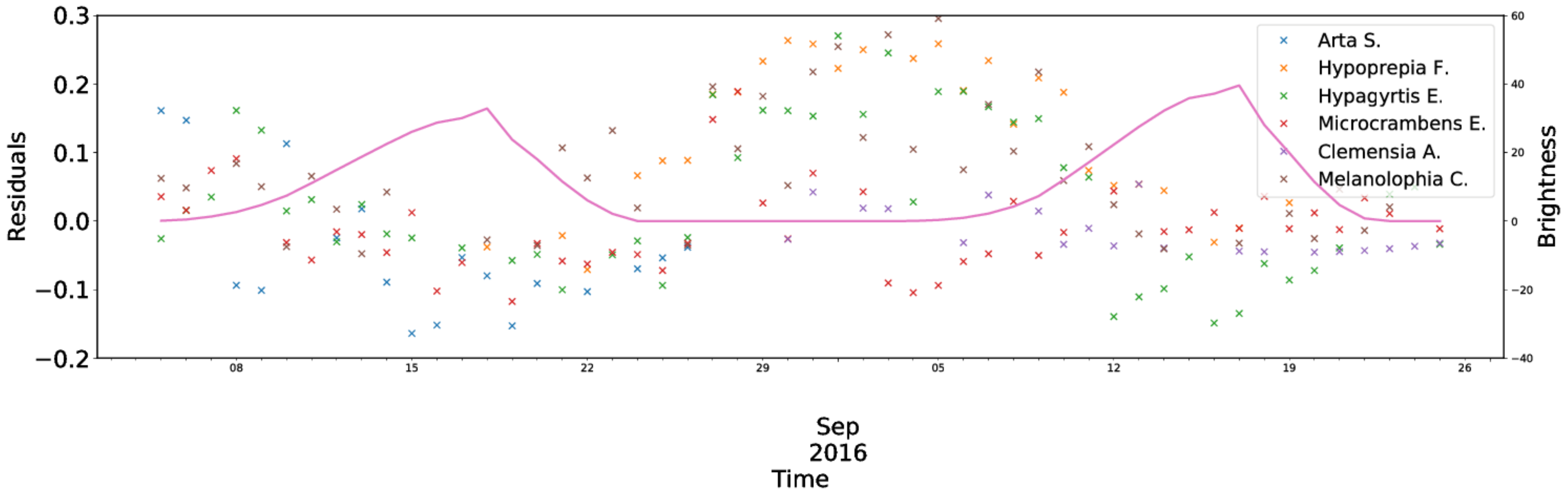}
    \caption{\label{fig:moth_resid_corr} { \small Residuals and moon brightness vs. time. Residuals are calculated relative to a predictive model fit across several years.}}
\vspace{-0.2cm}
\end{figure*}

\begin{figure}[t]
\centering
  \includegraphics[width=0.9\columnwidth]{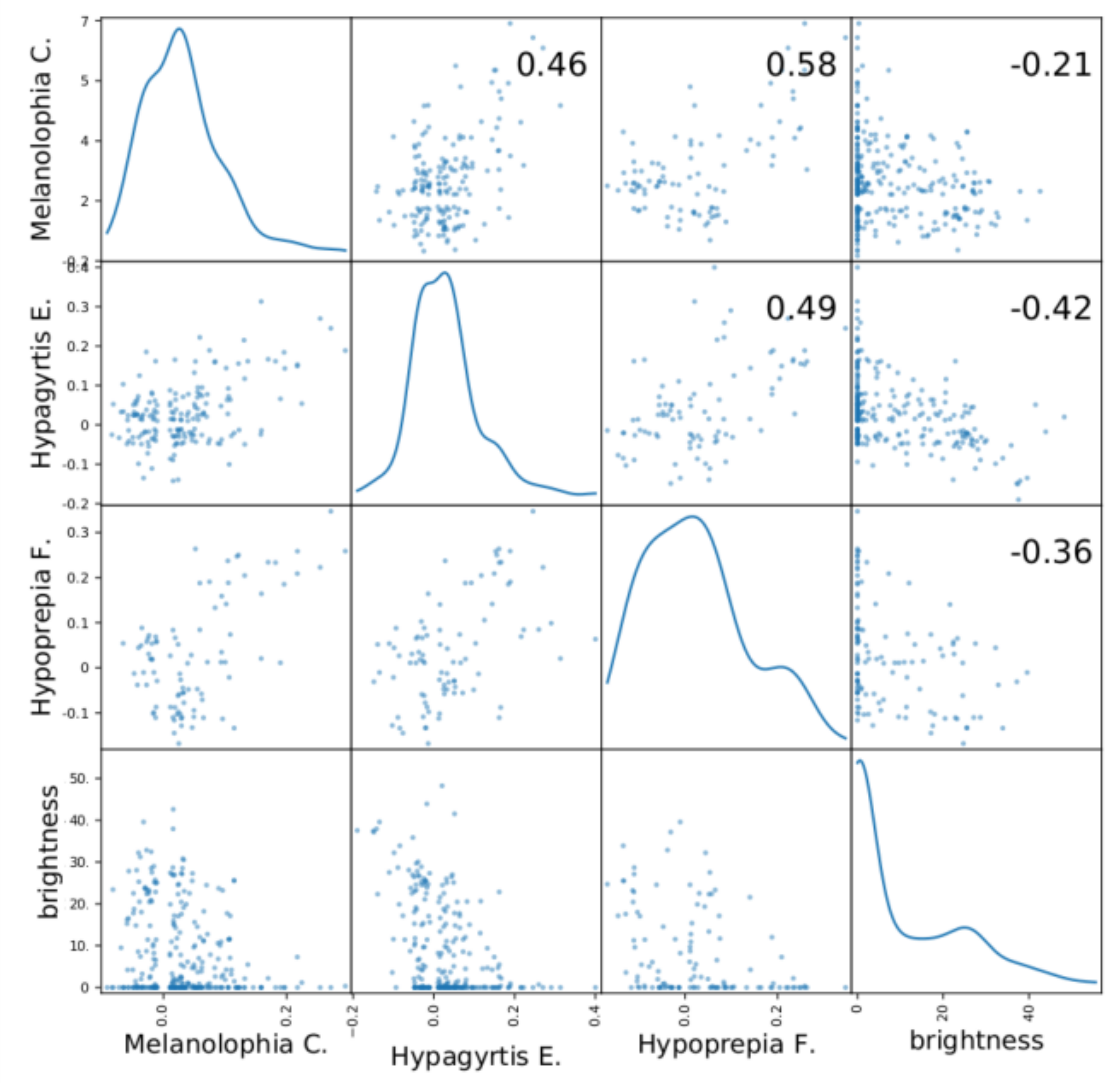}
  \caption{\label{fig:resid_scatter} {\small Pairwise scatter plots of residuals for common species and moon brightness. Inset numbers are correlations.}}
\end{figure}

Our second set of experiments use moth survey data from the Discover
Life project\footnote{https://www.discoverlife.org/moth} for studying spatio-temporal variation in moth communities. This dataset consists of counts of moths of different species collected at regular intervals at different study sites by both citizen scientists and moth experts. The protocol involves using artificial light to attract moths to a white surface and then photographing and identifying each specimen.
The dataset spans eight years and has been used for analysing seasonal flights \cite{pickering16a}, comparative taxonomy \cite{pickering16b} and other studies.
The data set is unique in its temporal resolution---at some of the sites, counts have been conducted almost every day for many years.

A typical use of this type of data is to estimate smooth curves of species abundance versus time such as the ones shown in Fig.~\ref{fig:moth-example} (see also~\cite{dennis2013indexing}). The smooth curve is then used as a proxy of the actual abundance at a given site in a given year to create indices of population size and to analyze temporal patterns, such as the timing of different generations (many species occur in discrete generations each year, e.g., see \emph{H. fucosa} in Fig.~\ref{fig:moth-example}). Scientists are especially interested in how populations and timing vary across years to understand population dynamics and potential links with external factors such as climate. 

\subsubsection{Moon Brightness}
A known systematic bias in moth counting studies is lunar brightness. On
nights when the moon is bright, moths are less likely to be attracted to the
white surface where they are photographed and counted.
We show evidence of this in Figs.~\ref{fig:moth_resid_corr} and~\ref{fig:resid_scatter}. Fig.~\ref{fig:moth_resid_corr} shows the residuals of moth counts with respect to a fitted model of count vs. day-of-year, together with moon brightness. The residuals track the moon phases and are anti-correlated with moon brightness. For example notice the significantly positive residual during early September which seems to rise exactly as lunar brightness reaches its nadir.
Fig.~\ref{fig:resid_scatter} shows pairwise scatter plots and correlation values for the residuals of three moth species together with moon brightness. All residuals are negatively correlated with moon brightness, and positively correlated with each other, with moon brightness being one contributing factor. In this example residuals are generally more correlated with each other than with moon brightness. One explanation is that moon brightness is measured using astronomical calculations and does not correspond exactly to brightness on the ground, which depends on factors like cloud cover. A second explanation is the effect of other factors besides brightness on detectability. A third (and less desirable) explanation is that there are other causal factors that are unmeasured. The reality is likely some combination of the three.
\subsubsection{Overall Experimental Strategy}
Our hypothesis is that moon brightness is one source of systematic
measurement error that can be effectively detected and removed using
3QS-regression. To test this, we will use moon brightness as an
external variable that is not available to our methods, and test the
correlation of moth counts with brightness before and after applying
denoising techniques, as well as the accuracy of predictive models on
``gold-standard'' test nights when lunar brightness is minimal,
so this particular source of variability is removed. 

Our overall goal is to develop the best smoothed models
of abundance over time for individual years
(cf. Fig.~\ref{fig:moth-example}), and we
hypothesize that correcting systematic errors will help.
Evaluation is challenging  because the true abundance
in a given year is unknown.
Our approach will be to fit smooth models in individual
years and test their predictive accuracy for \emph{other}
years. It is not our hope that the predictive accuracy is perfect ---
indeed, we \emph{want} our model to preserve variability across
years.
However, another source of such variability is differently aligned moon phases, which affect detection but not true abundance.
A model that corrects for such systematic
errors should therefore generalize better to other years
by reducing this source of variation.

\subsubsection{Methodology}
We choose moth counts from one site (Blue Heron Drive) for 2013
through 2018. Lunar brightness for the time of observation is computed with the software Pyephem\footnote{https://pypi.org/project/pyephem/}
using the latitude and longitude of the site.  We use the log-transformed counts of the most
common 10 species. 

We hold out one year at a time for testing and make predictions using
each other year, for a total of $20$ train-year / test-year
pairs. When predicting one test year, we first use all \emph{four}
other years to fit the regression models $\hat{\E}[Y_i | X]$ and
$\hat{\E}[R_i | R_{-i}]$ where $R_i = Y_i - \hat{\E}[Y_i | X]$ is the
residual. In all experiments, $X$ consists only of day-of-year.
We then compute $\hat{Z}_i$ for each training
year using the 3QS-estimator, fit a smoothed regression
model $\hat{\E}[\hat{Z}_i | X]$ to the denoised counts for a \emph{single} training year at a time, and use the
model to predict on the test year. This is repeated for all train-year / test-year pairs and for all species.
We compare 3QS-regression to multiple baselines,
including using the original measurements $Y_1$ to estimate $Z_1$, as well as the half-sibling (HS) regression  method of \cite{Scholkopf15}.
The
pyGAM package ~\cite{pyGAM} for generalized additive models (GAMs)
is used in all regression models.

\begin{table}[t]
  \small
  \centering
\begin{tabular}{|c|ccc|cc|}
  \hline
  & \multicolumn{3}{|c|}{Correlation with MB} &
  \multicolumn{2}{|c|}{Std. dev. \%} \\
Species & { $Y$} & {$\hat{Z}_{\text{HS}}$} & {$\hat{Z}_{\text{3QS}}$} & {$\hat{Z}_{\text{HS}}$} & {$\hat{Z}_{\text{3QS}}$}\\
\hline
\hline
\small Halysidota H. & -0.20 & -0.14 & -0.16 & 0.91 & 0.96\\
\small Hypoprepia F. & -0.15 & -0.01 & -0.04 & 0.80 & 0.90\\
\small Hypagyrtis E. & -0.18 & -0.09  & -0.11 &0.82 & 0.95\\
\small Microcrambus E.  & -0.16 & -0.03 & -0.10 & 0.84 & 0.97\\
\small Clemensia A. & -0.13 & -0.09 & -0.05 & 0.92 & 0.95\\
\small Lochmaeus B. & -0.13 & -0.01 & -0.05 &0.82 & 0.92\\
\small Melanolophia C. & -0.29 & -0.22 & -0.20 &0.77 &0.88\\
\small Iridopsis D. & -0.18 & -0.15 & -0.13 & 0.94 & 0.97\\
\hline
\end{tabular}
\caption{\label{tab:exp2} {\small Correlation with lunar brightness
  before and after denoising, and retained standard deviation as a fraction of that of $Y$}}
\end{table}

\begin{figure}[t]
  \begin{subfigure}[hb]{0.49\linewidth}
    \centering
    \includegraphics[width=\linewidth]{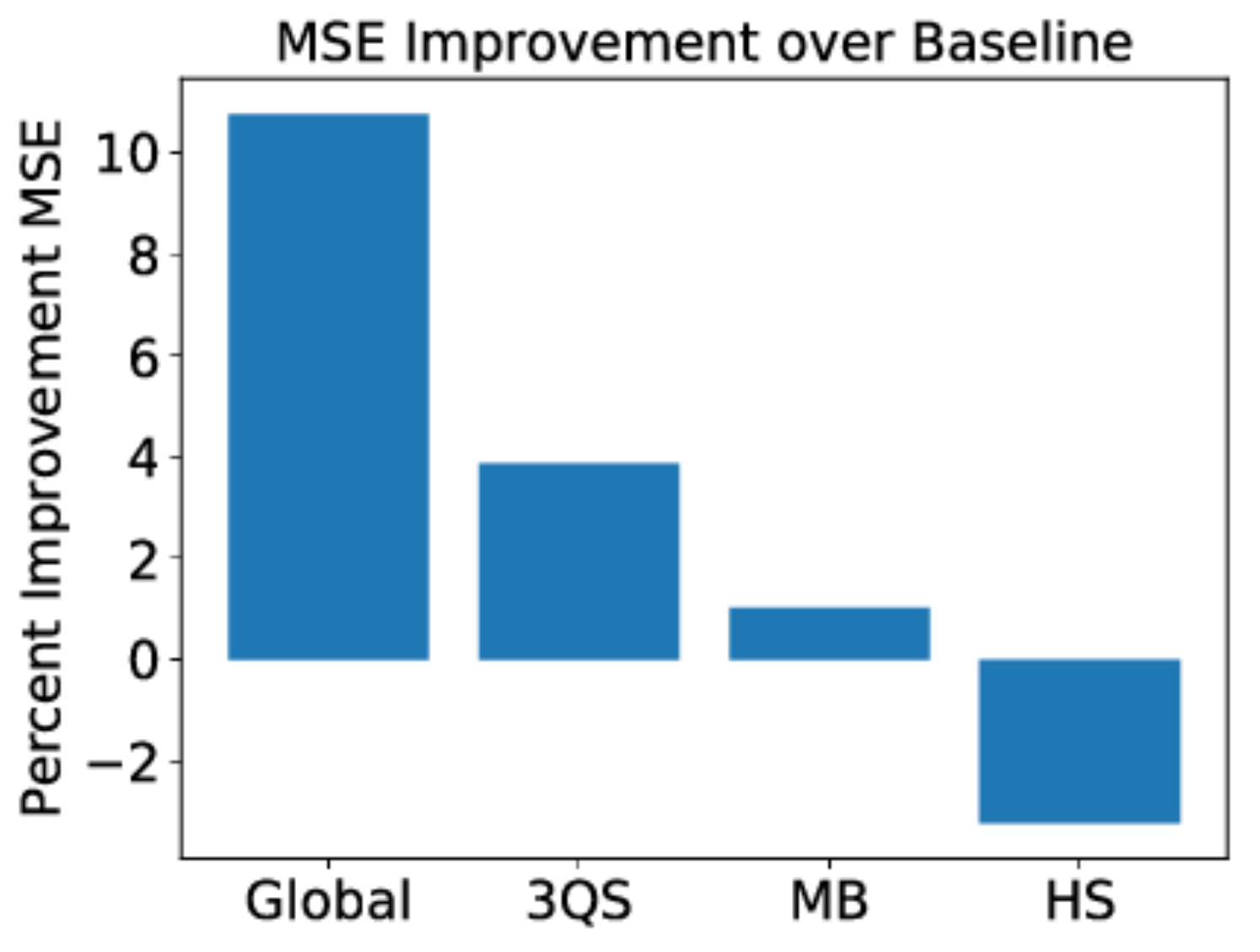}
    \caption{\label{fig:mse_bar}}
  \end{subfigure}
  \begin{subfigure}[hb]{0.47\linewidth}
    \centering
    \includegraphics[width=\linewidth]{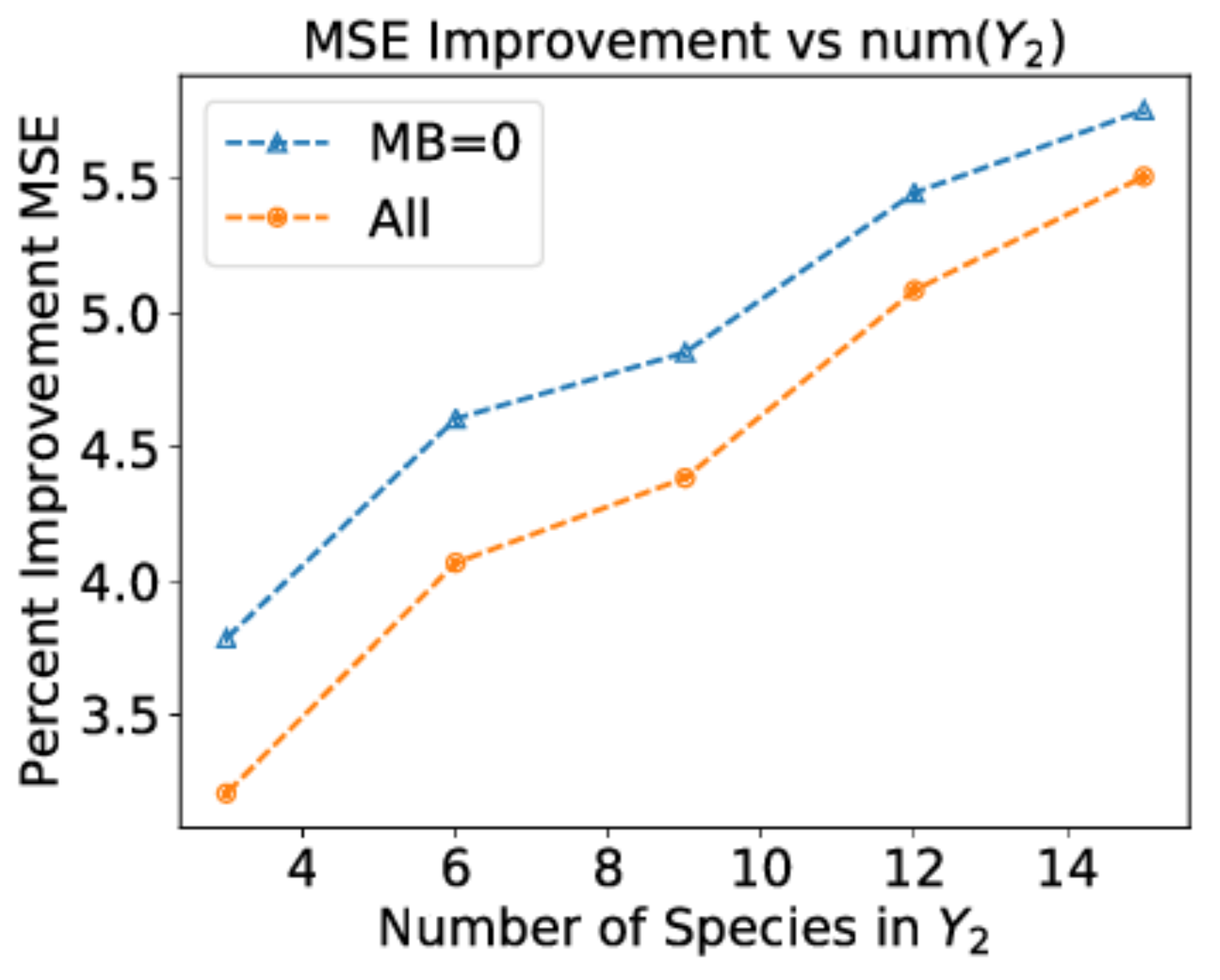}
    \caption{\label{fig:moth_err_vs_num}}
  \end{subfigure}
  \caption{Average percent improvement in predictive MSE relative to a GAM fitted to the raw counts. (a) Different estimators. (b) 3QS with increasing number of species.}
\end{figure}


\paragraph{Results: Moon Brightness.}
We first compare the relationship between counts and moon brightness
before and after 3QS-regression.  
It is expected that most species will have negative correlation
with lunar brightness initially, but the correlation will be significantly less after denoising.
The results are presented in Table
\ref{tab:exp2}.
Typical correlations between raw counts and lunar
brightness are on the order of -15\% to -20\%;
after denosing, the magnitude of correlations decrease by an average
of about 7\%.
The HS-regression estimator is shown for comparison, and also
decorrelates the counts from lunar brightness.
However if counts are correlated due to a common cause, HS-regression
will also remove some of the intrinsic (non-noise) variability (the
``throwing out the baby with the bathwater'' problem).
This can also be seen in Table~\ref{tab:exp2}: the overall variance of
$\hat{Z}_{\text{HS}}$ is significantly reduced relative to the raw counts.
We will show in our next experiment this removed variance corresponds to true variability that should not be removed.


\paragraph{Results: Predictive Accuracy.}
Fig.~\ref{fig:mse_bar} shows predictive accuracy of several methods averaged over all species and train-year
/ test-year pairs.
 The numbers are percent improvement in MSE relative
to the baseline of a GAM fit to the noisy counts in the training
year. MSE is computed only on
data from the test-year with moon brightness zero.
``Global'' is an oracle model shown for comparison. It is trained on raw counts of four training years instead of one. By fitting one model to multiple training years, it is expected to smooth out both sources of year-to-year variability (intrinsic and moon phase) and therefore predict better on a held-out year. However, this is \emph{not} our modeling goal --- we want to \emph{preserve} intrinsic year-to-year variability and eliminate variability due to moon phase.
``MB'' is a model that includes moon brightness as a feature to model detection variability.
The results show that 3QS outperforms all competitors, including the
MB model that has access to moon brightness as a feature. The global
model performs better, as expected. The HS model is worse than the
baseline due to removing intrinsic variability. Fig.~\ref{fig:moth_err_vs_num} shows the impact of using a greater number of species for $Y_2$. MSE is reported for the full test set as well as the moon-brightness zero set.
By Theorem~\ref{thm:additive}, we expect more species to reduce the conditional
variance and therefore improve accuracy, which is borne out in
Fig.~\ref{fig:moth_err_vs_num} on both test-sets.
Accuracy on the full test set is worse than on the moon-brightness
zero test set, which is expected due to extra variability in the test set.

\section{Conclusion}
%
Building on recent work on causal models, we presented three-quarter-sibling (3QS) regression to remove systematic measurement error for species that share observed common causes. 
We theoretically analyzed 3QS and presented empirical evidence of its success.
%
A future line of work would be to use simultaneous measurements to quantify uncertainty, possibly combined with Bayesian methods~\cite{Ellison2004}.
Our method may also be useful applied to other domains, such as under-reporting of drug use~\cite{adams2019learning}. A final question would be to address confounding in the presence of non-independent noise, i.e., when $X$ is not independent of $N$.

\clearpage

\bibliographystyle{named}
\bibliography{sibling}

\begin{thebibliography}{}

\bibitem[\protect\citeauthoryear{Adams \bgroup \em et al.\egroup
  }{2019}]{adams2019learning}
Roy Adams, Yuelong Ji, Xiaobin Wang, and Suchi Saria.
\newblock Learning models from data with measurement error: Tackling
  underreporting.
\newblock {\em arXiv preprint arXiv:1901.09060}, pages 1--8, 2019.

\bibitem[\protect\citeauthoryear{Dail and Madsen}{2011}]{Dail2011}
David Dail and Lisa Madsen.
\newblock {Models for estimating abundance from repeated counts of an open
  metapopulation.}
\newblock {\em Biometrics}, 67(2):577--87, 2011.

\bibitem[\protect\citeauthoryear{Dennis \bgroup \em et al.\egroup
  }{2013}]{dennis2013indexing}
Emily~B. Dennis, Stephen~N. Freeman, Tom Brereton, and David~B. Roy.
\newblock Indexing butterfly abundance whilst accounting for missing counts and
  variability in seasonal pattern.
\newblock {\em Methods in Ecology and Evolution}, 4(7):637--645, 2013.

\bibitem[\protect\citeauthoryear{Ellison}{2004}]{Ellison2004}
Aaron Ellison.
\newblock Bayesian inference in ecology.
\newblock {\em Ecology Letters}, 7(6):509--52, 2004.

\bibitem[\protect\citeauthoryear{Fr{\'e}nay and
  Verleysen}{2014}]{frenay2014classification}
Beno{\^\i}t Fr{\'e}nay and Michel Verleysen.
\newblock Classification in the presence of label noise: a survey.
\newblock {\em IEEE transactions on neural networks and learning systems},
  25(5):845--869, 2014.

\bibitem[\protect\citeauthoryear{Hutchinson \bgroup \em et al.\egroup
  }{2017}]{hutchinson2017species}
Rebecca~A. Hutchinson, Liqiang He, and Sarah~C Emerson.
\newblock Species distribution modeling of citizen science data as a
  classification problem with class-conditional noise.
\newblock In {\em AAAI}, pages 4516--4523, 2017.

\bibitem[\protect\citeauthoryear{Kelling \bgroup \em et al.\egroup
  }{2015}]{kelling2015can}
Steve Kelling, Alison Johnston, Wesley~M Hochachka, Marshall Iliff, Daniel
  Fink, Jeff Gerbracht, Carl Lagoze, Frank~A La~Sorte, Travis Moore, Andrea
  Wiggins, et~al.
\newblock Can observation skills of citizen scientists be estimated using
  species accumulation curves?
\newblock {\em PLoS One}, 10(10):e0139600, 2015.

\bibitem[\protect\citeauthoryear{Knape and Korner-Nievergelt}{2015}]{Knape2015}
Jonas Knape and Fr{\"a}nzi Korner-Nievergelt.
\newblock Estimates from non-replicated population surveys rely on critical
  assumptions.
\newblock {\em Methods in Ecology and Evolution}, 6(3):298--306, 2015.

\bibitem[\protect\citeauthoryear{Knape and
  Korner-Nievergelt}{2016}]{knape2016assumptions}
Jonas Knape and Fr{\"a}nzi Korner-Nievergelt.
\newblock On assumptions behind estimates of abundance from counts at multiple
  sites.
\newblock {\em Methods in Ecology and Evolution}, 7(2):206--209, 2016.

\bibitem[\protect\citeauthoryear{Lele \bgroup \em et al.\egroup
  }{2012}]{lele2012dealing}
Subhash~R. Lele, Monica Moreno, and Erin Bayne.
\newblock Dealing with detection error in site occupancy surveys: what can we
  do with a single survey?
\newblock {\em Journal of Plant Ecology}, 5(1):22--31, 2012.

\bibitem[\protect\citeauthoryear{MacKenzie \bgroup \em et al.\egroup
  }{2002}]{MacKenzie2002}
Darryl~I. MacKenzie, James~D. Nichols, Gideon~B. Lachman, Sam Droege, Andrew
  Royle, and Catherine~A. Langtimm.
\newblock {Estimating site occupancy rates when detection probabilities are
  less than one}.
\newblock {\em Ecology}, 83(8):2248--2255, 2002.

\bibitem[\protect\citeauthoryear{Nettleton \bgroup \em et al.\egroup
  }{2010}]{nettleton2010study}
David~F. Nettleton, Albert Orriols-Puig, and Albert Fornells.
\newblock A study of the effect of different types of noise on the precision of
  supervised learning techniques.
\newblock {\em Artificial intelligence review}, 33(4):275--306, 2010.

\bibitem[\protect\citeauthoryear{Pickering \bgroup \em et al.\egroup
  }{2016}]{pickering16b}
John Pickering, Dorothy Madamba, Tori Staples, and Rebecca Walcott.
\newblock Status of moth diversity and taxonomy: a comparison between {Africa}
  and {North America} north of {Mexico}.
\newblock {\em S. Lep. News}, 38(3):241--248, 2016.

\bibitem[\protect\citeauthoryear{Pickering}{2015}]{pickering16a}
John Pickering.
\newblock Why fly now? {P}upa banks, aposematism, and other factors that may
  explain observed moth flight activity.
\newblock {\em S. Lep. News}, 38(1):67--72, 2015.

\bibitem[\protect\citeauthoryear{Royle}{2004}]{Royle2004}
J.~Andrew Royle.
\newblock {N}-mixture models for estimating population size from spatially
  replicated counts.
\newblock {\em Biometrics}, 60(1):108--115, 2004.

\bibitem[\protect\citeauthoryear{Sch\"{o}lkopf \bgroup \em et al.\egroup
  }{2015}]{Scholkopf15}
Bernhard Sch\"{o}lkopf, David Hogg, Dun Wang, Dan Foreman-Mackey, Dominik
  Janzing, Carl-Johann Simon-Gabriel, and Jonas Peters.
\newblock Removing systematic errors for exoplanet search via latent causes.
\newblock In {\em Proceedings of the 32nd International Conference on Machine
  Learning (ICML)}, pages 2218--2226. PMLR, 2015.

\bibitem[\protect\citeauthoryear{Servén and Brummitt}{2018}]{pyGAM}
Daniel Servén and Charlie Brummitt.
\newblock {pyGAM}: Generalized additive models in {Python}, March 2018.

\bibitem[\protect\citeauthoryear{S{\'o}lymos and
  Lele}{2016}]{solymos2016revisiting}
P{\'e}ter S{\'o}lymos and Subhash~R. Lele.
\newblock Revisiting resource selection probability functions and single-visit
  methods: clarification and extensions.
\newblock {\em Methods in Ecology and Evolution}, 7(2):196--205, 2016.

\bibitem[\protect\citeauthoryear{Yu \bgroup \em et al.\egroup
  }{2010}]{yu2010modeling}
Jun Yu, Weng-Keen Wong, and Rebecca~A. Hutchinson.
\newblock Modeling experts and novices in citizen science data for species
  distribution modeling.
\newblock In {\em 2010 IEEE International Conference on Data Mining (ICDM)},
  pages 1157--1162. IEEE, 2010.

\bibitem[\protect\citeauthoryear{Yu \bgroup \em et al.\egroup
  }{2014a}]{yu2014latent}
Jun Yu, Rebecca~A. Hutchinson, and Weng-Keen Wong.
\newblock A latent variable model for discovering bird species commonly
  misidentified by citizen scientists.
\newblock In {\em Twenty-Eighth AAAI Conference on Artificial Intelligence},
  2014.

\bibitem[\protect\citeauthoryear{Yu \bgroup \em et al.\egroup
  }{2014b}]{yu2014clustering}
Jun Yu, Weng-Keen Wong, and Steve Kelling.
\newblock Clustering species accumulation curves to identify skill levels of
  citizen scientists participating in the {eBird} project.
\newblock In {\em Twenty-Sixth IAAI Conference}, 2014.

\end{thebibliography}




\end{document}